\documentclass{article}
\usepackage{algorithm}
\usepackage{algpseudocode}
\usepackage[utf8]{inputenc}
\usepackage{amsmath}
\usepackage{amsfonts, mathtools}
\usepackage{amssymb}
\usepackage{enumitem}
\usepackage[left=1.5cm, right=2cm, top = 2cm, bottom=2cm]{geometry}
\usepackage[numbers]{natbib}
\usepackage{lineno,hyperref}
\usepackage{amsmath,amssymb,amsthm}
\usepackage{mathrsfs} %calligraphy fonts
\usepackage{graphicx, tikz}
\usepackage[caption=false]{subfig}
\usepackage{wrapfig}
\usepackage{tikz}
\usepackage{pgfplots}
\usetikzlibrary{arrows.meta}
\pgfplotsset{compat=1.16}
\newtheorem{theorem}{Theorem}

\newtheorem{lemma}{Lemma}
\newtheorem{definition}{Definition}
\newtheorem{remark}{Remark}

\DeclareMathAlphabet{\mathpzc}{OT1}{pzc}{m}{it}%Styled letters for sockets

\modulolinenumbers[5]

\newcommand{\ta}[1]{}
\newcommand{\hs}[1]{}

\title{A Fast Algorithm for Finding Minimum Weight Cycles in Mining Cyclic Graph Topologies}

\author{Heman Shakeri\textsuperscript{1,*}  \and Torben Amtoft\textsuperscript{3} \and Behnaz Moradi-Jamei\textsuperscript{2} \and Nathan Albin\textsuperscript{4} \and  Pietro Poggi-Corradini\textsuperscript{4}}
\date{}

\begin{document}
\maketitle

\begin{abstract}
Cyclic structures are fundamental topological features in graphs, playing critical roles in network robustness, information flow, community structure, and various dynamic processes. Algorithmic tools that can efficiently probe and analyze these cyclic topologies are increasingly vital for tasks in graph mining, network optimization, bioinformatics, and social network analysis. A core primitive for quantitative analysis of cycles is finding the Minimum Weight Cycle (MWC), representing the shortest cyclic path in a weighted graph. However, computing the MWC efficiently remains a challenge, particularly compared to shortest path computations.
This paper introduces a novel deterministic algorithm for finding the MWC in general weighted graphs (\href{https://github.com/Shakeri-Lab/girth}{github.com/Shakeri-Lab/girth}). Our approach adapts the structure of Dijkstra's algorithm by introducing and minimizing a \textit{composite distance} metric, effectively translating the global cycle search into an iterative node-centric optimization. We provide a rigorous proof of correctness based on loop invariants.
We detail two mechanisms for accelerating the search: a provable node discarding technique based on intermediate results, and a highly effective graph pruning heuristic. This heuristic dynamically restricts the search to relevant subgraphs, leveraging the principle of locality often present in complex networks to achieve significant empirical speedups, while periodic resets ensure global optimality is maintained.
The efficiency of the proposed MWC algorithm enables its use as a core component in more complex analyses focused on cyclic properties. We illustrate this through a detailed application case study: accelerating the computation of the Loop Modulus, a measure of cycle richness used in advanced network characterization (\href{https://github.com/Shakeri-Lab/loop\_modulus}{github.com/Shakeri-Lab/loop\_modulus}). Our algorithm dramatically reduces the runtime of the iterative constraint-finding bottleneck in this computation.

\end{abstract}

\noindent\textsuperscript{1}School of Data Science, University of Virginia\\
\textsuperscript{2}Department of Mathematics and Statistics, University of Kansas\\
\textsuperscript{3}Department of Computer Science, Kansas State University\\
\textsuperscript{4}Department of Mathematics, Kansas State University\\
\textsuperscript{*}Corresponding author: \href{mailto:hs9hd@virginia.edu}{hs9hd@virginia.edu}

\section{Problem Statement and Framework}
\label{sec:intro_framework}

Cyclic structures permeate real-world networks and abstract graph representations, influencing phenomena ranging from network resilience and transport efficiency to the formation of communities and the behavior of dynamic systems \cite{diestel2000graph}. Effectively analyzing these cyclic topologies requires robust algorithmic primitives. A cornerstone primitive is the computation of the Minimum Weight Cycle (MWC), often referred to as the (weighted) girth of the graph \cite{harary1969graph}. Finding the MWC is fundamental not only in graph theory itself but also serves as a building block for applications such as computing minimal cycle bases \cite{gleiss2001circuit, kavitha2004faster, kavitha2007new}, analyzing cycle packing problems \cite{caprara2003packing, krivelevich2005approximation, salavatipour2005disjoint}, understanding graph connectivity and chromatic properties \cite{diestel2000graph, djidjev2000computing}, and even solving related problems like min-cut in planar graph duals \cite{sankowski2011min}. Furthermore, efficient MWC computation is crucial for network analysis techniques like Loop Modulus, which quantifies the richness of cycle families \cite{shakeri2016modulus}.

Despite its importance, finding the MWC in general weighted undirected graphs algorithmically is considerably more challenging than finding shortest paths between vertex pairs. Let $G=(V, E, w)$ be an undirected graph with vertex set $V$, edge set $E$, and a positive edge weight function $w: E \to \mathbb{R}_{> 0}$. A \textit{path} $\pi$ is a sequence of distinct vertices $v_0, v_1, \dots, v_r$ such that $(v_{i-1}, v_i) \in E$ for all $i=1, \dots, r$. A \textit{simple cycle} $c$ is a sequence of vertices $v_0, v_1, \dots, v_r$ where $v_0 = v_r$, $r \ge 3$, and $v_1, \dots, v_r$ are distinct. We denote the set of vertices in a path or cycle $\pi$ as $\text{vertices}(\pi)$. The \textit{length} of a path $\pi = (v_0, \dots, v_r)$ or cycle $c = (v_0, \dots, v_r)$ is the sum of its edge weights:
\begin{equation}
\label{eq:length_def}
\ell(\pi) := \sum_{i=1}^r w(v_{i-1}, v_i).
\end{equation}
The \textit{shortest path distance} $d(x, y)$ between vertices $x, y \in V$ is the minimum length $\ell(\pi)$ over all paths $\pi$ from $x$ to $y$. If no path exists, $d(x, y) = \infty$. The Minimum Weight Cycle (MWC) problem seeks to find a simple cycle $c^*$ in $G$ such that $\ell(c^*) = \min_{c \in \mathscr{C}} \ell(c)$, where $\mathscr{C}$ is the set of all simple cycles in $G$. We assume $G$ is not a tree (or a forest), so $\mathscr{C}$ is non-empty.

\paragraph{Existing Approaches and Challenges}
Significant effort has been invested in developing algorithms for the MWC problem. For unweighted graphs, Itai and Rodeh \cite{itai1978finding} presented algorithms with subcubic time complexity related to matrix multiplication, though the weighted case was left open. Subsequent work extended these ideas to integer weights \cite{roditty2011minimum} and established connections between MWC and other challenging graph problems \cite{williams2010subcubic}. Approximation algorithms \cite{itai1978finding, Roditty20111446, lingas2009efficient, yuster2011shortest, peleg2012distributed}, randomized algorithms \cite{yuster2011shortest}, and specialized algorithms for finding even-length cycles \cite{yuster1997finding} have also been developed.
A common strategy employed by many deterministic MWC algorithms for weighted graphs involves searching for the shortest cycle passing through a specific vertex or edge.
The standard baseline, often called the \textbf{Rooted Girth Algorithm}, iterates through every edge $e=(u,v) \in E$. For each edge, it computes the shortest path distance $d_{G\setminus e}(u,v)$ in the graph without edge $e$. The minimum of $d_{G\setminus e}(u,v) + w(e)$ over all edges is the MWC. This requires $|E|$ shortest path computations.

Alternatively, one can adapt all-pairs shortest path algorithms or run multiple shortest path searches (like Dijkstra's) from each vertex $v$, looking for the shortest path between neighbors $u, z$ of $v$ in the graph $G \setminus \{v\}$. Using efficient implementations (e.g., Fibonacci heaps), this leads to overall complexities like $O(|V||E| + |V|^2 \log |V|)$ or related bounds \cite{itai1978finding, yuster2011shortest, orlin2016nm}.
While effective, these approaches repeatedly explore large portions of the graph.

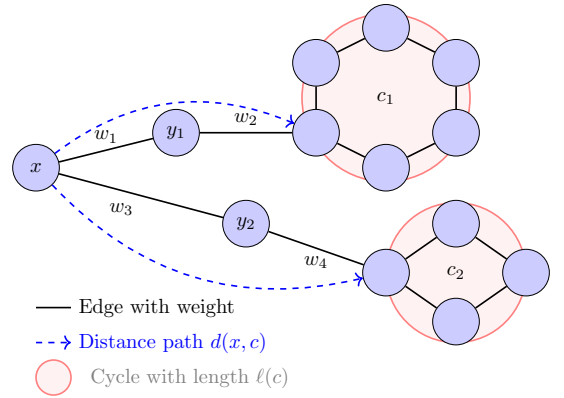
\begin{wrapfigure}[21]{r}{0.4\textwidth}
\vspace{-14pt}
	\centering
    \resizebox{0.4\textwidth}{!}{% figures/composite_distance.tex - TikZ content only
\begin{tikzpicture}[
    scale=1.2,
    vertex/.style={circle, draw, fill=blue!20, minimum size=8mm},
    edge/.style={thick, -},
    cycle/.style={thick, draw=red, fill=red!10, opacity=0.5},
    distancepath/.style={->, dashed, thick, blue},
    label/.style={font=\small}
]
    % Define the primary graph structure
    \node[vertex] (x) at (0,0) {$x$};
    
    % Define the cycles with filled areas
    \filldraw[cycle] (5,1) circle (1.2cm);   % Cycle c1
    \filldraw[cycle] (6,-1.5) circle (1cm);  % Cycle c2
    
    % Add cycle vertices for c1
    \node[vertex] (c1) at (4,1.5) {};
    \node[vertex] (c2) at (5,2) {};
    \node[vertex] (c3) at (6,1.5) {};
    \node[vertex] (c4) at (6,0.5) {};
    \node[vertex] (c5) at (5,0) {};
    \node[vertex] (c6) at (4,0.5) {};
    
    % Add cycle vertices for c2
    \node[vertex] (d1) at (5,-1.5) {};
    \node[vertex] (d2) at (6,-0.8) {};
    \node[vertex] (d3) at (7,-1.5) {};
    \node[vertex] (d4) at (6,-2.2) {};
    
    % Add intermediate vertices
    \node[vertex] (y1) at (2,0.5) {$y_1$};
    \node[vertex] (y2) at (3,-0.8) {$y_2$};
    
    % Add connecting edges with weight labels
    \draw[edge] (x) -- (y1) node[midway, above] {$w_1$};
    \draw[edge] (y1) -- (c6) node[midway, above] {$w_2$};
    \draw[edge] (x) -- (y2) node[midway, below left] {$w_3$};
    \draw[edge] (y2) -- (d1) node[midway, below] {$w_4$};
    
    % Add cycle edges for c1
    \draw[edge] (c1) -- (c2) -- (c3) -- (c4) -- (c5) -- (c6) -- (c1) -- cycle;
    
    % Add cycle edges for c2
    \draw[edge] (d1) -- (d2) -- (d3) -- (d4) -- (d1) -- cycle;
    
    % Draw distance paths from x to cycles
    \draw[distancepath] (x) to[bend left] (c6);  % Path to c1 via c6
    \draw[distancepath] (x) to[bend right] (d1); % Path to c2 via d1
    
    % Add labels for cycles
    \node at (5,1) {$c_1$};
    \node at (6,-1.5) {$c_2$};
    
    % Add a legend, shifted to avoid overlap
    \begin{scope}[shift={(0,-2)}]
        \draw[edge] (0,0) -- (0.5,0) node[right] {Edge with weight};
        \draw[distancepath] (0,-0.5) -- (0.5,-0.5) node[right] {Distance path $d(x,c)$};
        \draw[cycle, fill=red!10] (0.25,-1) circle (0.25) node[right=0.5cm] {Cycle with length $\ell(c)$};
    \end{scope}
\end{tikzpicture}}
    \vspace{-13pt}
	\caption{Composite distance  $d^+(x,c)$ is the shortest path distance from vertex $x$ to cycle $c$ plus $c$'s length $\ell(c)$. The plot shows cycles $c_1$ and $c_2$ with edge weights. The Minimum Weight Cycle (MWC) length is $\ell(c^*) = \min_{x \in V} \min_{c \in \mathscr{C}} d^+(x,c)$.}
    \label{fig:composite_distance}
\end{wrapfigure}

\subsection{An Alternative View: Composite Distance Minimization}
\label{subsec:composite_distance}

In contrast to vertex- or edge-rooted searches, this paper introduces a different perspective based on minimizing a metric that couples shortest path distances with cycle lengths. This allows for a vertex-centric iterative approach inspired by Dijkstra's algorithm.

\begin{definition}[Composite Distance]\label{def:composite_distance}
For a vertex $x \in V$ and a cycle $c \in \mathscr{C}$, the distance from $x$ to $c$ is $d(x,c) = \min_{y \in \text{vertices}(c)} d(x,y)$. The \textit{composite distance} from vertex $x$ to cycle $c$ is defined as:
\begin{equation}\label{eq:composite}
d^+(x,c) = d(x,c) + \ell(c).
\end{equation}
The composite distance from vertex $x$ to the family of all cycles $\mathscr{C}$ is:
\begin{equation}\label{eq:closestC}
d^+(x) = \min_{c\in \mathscr{C}} d^+(x,c).
\end{equation}
\end{definition}

This definition shifts the focus from finding a cycle directly to finding a vertex $x$ that minimizes this composite distance $d^+(x)$. The following theorem establishes the equivalence between minimizing $d^+(x)$ and finding the MWC.

\begin{theorem}[Equivalence to MWC]\label{thm:minmin_equiv}
	Minimizing $d^+(x)$ over $x\in V$ is equivalent to finding the length of the minimum weight cycle in $G$. Moreover, the minimum value $\min_{x \in V} d^+(x)$ equals $\ell(c^*)$ for any MWC $c^*$, and this minimum is attained for any vertex $x \in \text{vertices}(c^*)$.
\end{theorem}
\begin{proof}
	For any $x\in V$ and $c\in \mathscr{C}$, $d^+(x,c) = d(x,c) + \ell(c) \ge \ell(c)$ since $d(x,c) \ge 0$. Equality $d^+(x,c) = \ell(c)$ holds if and only if $d(x,c) = 0$, i.e. $x \in \text{vertices}(c)$.
     Let $c^*$ be a Minimum Weight Cycle. For any vertex $x^* \in \text{vertices}(c^*)$, the distance to the cycle is zero, so its composite distance is $d^+(x^*, c^*) = \ell(c^*)$. For any arbitrary vertex $y \in V$ and any cycle $c \in \mathscr{C}$, the composite distance $d^+(y,c) = d(y,c) + \ell(c) \ge \ell(c) \ge \ell(c^*)$.
\end{proof}

This shows that the global minimum value of the composite distance is precisely $\ell(c^*)$. Therefore, the overall minimum $\min_{x\in V} d^+(x)$ is equal to $\ell(c^*)$, and this minimum is achieved for any vertex on a Minimum Weight Cycle.
In other words, the MWC length (girth) of the graph is found by solving the following nested minimization problem (depicted in Figure~\ref{fig:composite_distance}):
\begin{equation}\label{eq_girth_thm}
\ell(c^*) = \min_{x\in V} d^+(x) = \min_{x\in V}\left( \min_{c\in \mathscr{C}} d^+(x,c)\right).
\end{equation}
Therefore, the proposed algorithm enumerates the vertices $V=\{1, 2, \dots, N\}$ and iteratively minimizes $d^+(x)$ for each vertex $x$. Information gathered during the search for $d^+(j)$ can be used to accelerate the subsequent search for $d^+(j+1)$. Specifically, the current best MWC length estimate can serve as a cut-off. 
This acceleration can be taken a step further. The following theorem provides a criterion for discarding certain vertices from future consideration entirely.

\begin{theorem}[Vertex Discarding Criterion]\label{thm:exclude_vertices}
Suppose for some vertex $j \in V$, the minimum composite distance is $d^+(j) = d(j,c) + \ell(c)$ for some cycle $c \in \mathscr{C}$. If $c$ is not a shortest cycle (i.e., a $c'$ exists with $\ell(c') < \ell(c)$), then any vertex $z$ processed after $j$ satisfying $d(j,z) \le d(j,c)$ cannot belong to any $c^*$ that is a MWC. Such vertices $z$ can thus be safely discarded from subsequent searches for the global MWC.
\end{theorem}
\begin{proof}
	Let $c^*$ be any MWC, then for a non MWC $c$, we have $\ell(c^*) < \ell(c)$. By the definition of $d^+(j)$ (Equation \ref{eq:closestC}), we also know $d^+(j) \le d^+(j,c^*)$. The premise assumes $d^+(j,c)$ is the minimum found associated with $j$, so $d^+(j) = d^+(j,c)$. Combining these gives:
	\begin{align*}
	d(j,c) + \ell(c) = d^+(j,c) = d^+(j) &\le d^+(j,c^*) \\
	&= d(j,c^*) + \ell(c^*) \\
	&< d(j,c^*) + \ell(c) \quad (\text{since } \ell(c^*) < \ell(c)).
	\end{align*}
	Comparing the start and end of this inequality chain and cancelling $\ell(c)$ yields:
	$$
	d(j,c) < d(j,c^*).
	$$
	Now consider any vertex $z$ such that $d(j,z) \le d(j,c)$. If $z$ were part of the MWC $c^*$, then by the definition of $d(j,c^*)$ (as the minimum distance from $j$ to any vertex in $c^*$), we would have $d(j,c^*) \le d(j,z)$. Combining these inequalities gives $d(j,c^*) \le d(j,z) \le d(j,c)$. However, this contradicts the established result $d(j,c) < d(j,c^*)$. Therefore, the initial assumption must be false: $z$ cannot belong to $c^*$.
\end{proof}
This composite distance framework, incorporating Theorem \ref{thm:minmin_equiv} and Theorem \ref{thm:exclude_vertices}, forms the basis of the Dijkstra-like algorithm detailed in Section \ref{sec_alg}.
The remainder of this paper is organized as follows. Section \ref{sec:complexity} discusses the algorithm's complexity, introduces the graph pruning heuristic for practical performance enhancement, and provides illustrative examples. Section \ref{sec:loop_modulus_app} demonstrates the application of the algorithm to accelerate Loop Modulus computations. Finally, Section \ref{sec:conclusion} offers concluding remarks.

%==========%==========%==========%==========

\section{An Iterative Algorithm based on Composite Distance}
\label{sec_alg}

Building upon the composite distance framework, we now present a deterministic algorithm designed to find MWC by iteratively minimizing $d^+(x)$ for each vertex $x \in V$. The algorithm adapts the core mechanics of Dijkstra's shortest path algorithm. For each starting vertex $x$, it grows a shortest path tree, but incorporates specific checks to identify cycles and update the global minimum cycle length found so far. Crucially, it includes an optimization based on the current shortest cycle estimate to potentially terminate the inner search early, aiming to improve performance over exhaustive searches.

\subsection{Algorithm Description}
\label{subsec:algo_desc}
The algorithm proceeds with an outer loop iterating through all potential starting vertices $x \in V$. The globally shortest cycle length found across all iterations is stored. Inside this loop, a modified Dijkstra-like search is performed starting from $x$.

\paragraph{Data Structures} The primary data structure maintained across the outer loop iterations is:
\begin{itemize}
    \item $\gamma$: Records the length of the shortest simple cycle found globally so far. Initially $\gamma = \infty$. It is updated whenever a shorter cycle is discovered. (The algorithm can be easily modified to store the cycle itself, not just its length).
    \item $V_{\text{active}}$: The set of vertices that have not been discarded and are still candidates for belonging to the MWC. Initially $V_{\text{active}} = V$.
\end{itemize}
The inner loop (the modified Dijkstra search starting from $x$) utilizes the following:
\begin{itemize}
    \item $Q$: The set of vertices for which the shortest path distance from $x$ has been finalized in the current inner iteration. Initially $Q = \emptyset$.
    \item $\delta$: A function mapping $V \to \mathbb{R}_{\ge 0} \cup \{\infty\}$. $\delta(y)$ stores the current upper bound on the shortest path distance $d(x,y)$. If $y \in Q$, then $\delta(y) = d(x,y)$. Initially, $\delta(x) = 0$ and $\delta(y) = \infty$ for $y \neq x$.
    \item $\texttt{pred}$: A partial function mapping vertices $v \in Q \setminus \{x\}$ to their predecessor $u = \texttt{pred}(v)$ on the shortest path from $x$ found so far. $u$ must be in $Q$. We use $\texttt{pred}^k$ for $k$ applications, $\texttt{pred}^+(v)$ for the set of all predecessors (excluding $v$), and $\texttt{pred}^*(v)$ for $\texttt{pred}^+(v) \cup \{v\}$.
    \item $d^+_{\min}$, $d_{\text{to\_cycle}}$, $\ell_{\text{best}}$: Variables tracking the minimum composite distance found from the current root $x$, the distance to the corresponding cycle, and its length.
\end{itemize}

The detailed procedure is presented in Algorithm \ref{alg:MWC_Dijkstra}. 

\paragraph{Computing the Distance to a Cycle}
A key quantity needed for the vertex discarding optimization is the distance $d(x,c)$ from the current root $x$ to a detected cycle $c$. When a cycle is detected via the edge $(y,z)$, the cycle $c$ consists of the paths in the shortest path tree from the LCA $p$ to $y$ and $z$, plus the edge $(y,z)$.

Since $p$ lies on the shortest path from $x$ to both $y$ and $z$, the closest vertex on cycle $c$ to the root $x$ is the LCA $p$ itself. Therefore:
\begin{equation}\label{eq:dist_to_cycle}
d(x,c) = \delta(p).
\end{equation}
This observation is crucial: we can compute the distance to any detected cycle directly from the $\delta$ value of its LCA.

The \textit{composite distance} from $x$ to cycle $c$ (Definition~\ref{def:composite_distance}) then becomes:
\begin{equation}\label{eq:composite_explicit}
d^+(x,c) = d(x,c) + \ell(c) = \delta(p) + [\delta(y) + \delta(z) + w(y,z) - 2\delta(p)] = \delta(y) + \delta(z) + w(y,z) - \delta(p).
\end{equation}

This algorithm augments the standard Dijkstra approach in three key ways:
\begin{enumerate}
    \item \textbf{Cycle Detection and Composite Distance Tracking:} When considering an edge $(y, z)$ where $y$ is newly added to $Q$ and $z$ is already in $Q$, a potential cycle is formed if $z$ is not the direct predecessor of $y$. The algorithm calculates both the cycle length $\ell(c)$ and the distance $d(x,c) = \delta(p)$ to this cycle, where $p$ is the Lowest Common Ancestor\footnote{The efficient computation of the LCA is a standard problem in algorithmic graph theory \cite{OnlineLCA} and can be handled using various techniques, often with logarithmic or near-constant time complexity per query after some preprocessing.} of $y$ and $z$ in the current shortest path tree rooted at $x$. The algorithm tracks the cycle achieving the minimum composite distance $d^+(x,c)$ found so far (Lines 24-26 in Algorithm \ref{alg:MWC_Dijkstra}) and updates the global minimum cycle length $\gamma$ (Line 23).
    
    \item \textbf{Vertex Discarding:} After the inner loop completes, the algorithm applies the criterion from Theorem~\ref{thm:exclude_vertices}. This is only applied if the cycle $c$ achieving $d^+_{\min}$ is not the MWC ($\ell_{\text{best}} > \gamma$) AND if we can guarantee $d^+_{\min}$ is the true $d^+(x)$ despite the $\gamma/2$ cutoff. This guarantee is provided by the condition $d^+_{\min} < 3\gamma/2$ (Line 30). If satisfied, any unprocessed vertex $z$ satisfying $\delta(z) \leq d_{\text{to\_cycle}}$ cannot belong to any MWC and is removed from $V_{\text{active}}$ (Lines 30-36).
    
    \item \textbf{Optimized Termination Condition:} The standard Dijkstra continues until all reachable vertices are in $Q$. This algorithm incorporates a check (Line 9): the inner `while' loop only continues exploring vertices $v \notin Q$ if their current distance estimate $\delta(v)$ is less than $\gamma/2$. This condition is crucial for efficiency and its validity is proven in the correctness analysis (Lemma \ref{lem:outer}). It stems from the observation that any potential shorter cycle involving an unexplored vertex $v$ must necessarily consist of two paths from $x$ to $v$ (one potentially being just the edge closing the cycle), and if $v$ is too far from $x$, it cannot contribute to a cycle shorter than $\gamma$.
\end{enumerate}

\begin{figure}[htbp]
    \centering
    \begin{minipage}{0.95\linewidth}
        \begin{algorithm}[H]
            \caption{MWC Algorithm based on Composite Distance Minimization}
            \label{alg:MWC_Dijkstra}
            \begin{algorithmic}[1]
            \algnotext{EndFor}
            \algnotext{EndIf}
            \algnotext{EndWhile}
                \State $\gamma \gets \infty$ \Comment{Shortest cycle length found so far}
                \State $V_{\text{active}} \gets V$ \Comment{Vertices not yet discarded}
                \ForAll{vertex $x \in V_{\text{active}}$}
                    \State Initialize $\delta(v) \gets \infty$ for all $v \in V$; $\texttt{pred}(v)$ undefined for all $v \in V$
                    \State $Q \leftarrow \emptyset$
                    \State $\delta(x) \leftarrow 0$
                    \State $d^+_{\min} \gets \infty$; $d_{\text{to\_cycle}} \gets \infty$; $\ell_{\text{best}} \gets \infty$ \Comment{Track composite distance}
                    \While{there exists $v \notin Q$ such that $\delta(v) < \gamma/2$} \label{line:while_condition}
                        \State $y \leftarrow \arg\min_{v \notin Q} \{\delta(v)\}$
                        \State $Q \leftarrow Q \cup \{y\}$
                        \ForAll{vertex $z$ adjacent to $y$}
                            \If{$z \notin Q$}
                                \If{$\delta(y) + w(y,z) < \delta(z)$}
                                    \State $\delta(z) \leftarrow \delta(y) + w(y,z)$
                                    \State $\texttt{pred}(z) \leftarrow y$
                                \EndIf
                            \ElsIf{$z \in Q$ and $z \neq \texttt{pred}(y)$} \label{line:cycle_check}
                                \State $p \leftarrow \text{LCA}_{\texttt{pred}^*}(y, z)$ \label{line:lca}
                                \State $\ell_c \leftarrow \delta(y) + \delta(z) + w(y,z) - 2\delta(p)$ \label{line:cycle_length_calc}
                                \State $d_{x,c} \gets \delta(p)$ \Comment{Distance from $x$ to cycle $c$} \label{line:dist_to_cycle}
                                \State $d^+_{x,c} \gets d_{x,c} + \ell_c$ \Comment{Composite distance} \label{line:composite_dist}
                                \State $\gamma \leftarrow \min(\gamma, \ell_c)$ \label{line:gamma_update}
                                \If{$d^+_{x,c} < d^+_{\min}$}
                                    \State $d^+_{\min} \gets d^+_{x,c}$; $d_{\text{to\_cycle}} \gets d_{x,c}$; $\ell_{\text{best}} \gets \ell_c$
                                \EndIf
                            \EndIf
                        \EndFor
                    \EndWhile
                    \State \Comment{Apply Vertex Discarding (Theorem~\ref{thm:exclude_vertices})}
                    \If{$d^+_{\min} < \infty$ \textbf{and} $\ell_{\text{best}} > \gamma$ \textbf{and} $d^+_{\min} < 3\gamma/2$} \label{line:discard_condition}
                        \ForAll{$z \in V_{\text{active}} \setminus \{x\}$ with $z$ not yet processed} \label{line:discard_loop}
                            \If{$\delta(z) \leq d_{\text{to\_cycle}}$} \Comment{$d(x,z) \leq d(x,c)$}
                                \State $V_{\text{active}} \gets V_{\text{active}} \setminus \{z\}$
                            \EndIf
                        \EndFor
                    \EndIf
                \EndFor
                \State \Return $\gamma$
            \end{algorithmic}
        \end{algorithm}
    \end{minipage}
\end{figure}

\subsection{Correctness Proof}
\label{subsec:correctness}

To establish that Algorithm \ref{alg:MWC_Dijkstra} correctly computes the length of the MWC, we rely on proving loop invariants for both the inner `while' loop (Lemma \ref{lem:inner}) and the outer `for' loop (Lemma \ref{lem:outer}). A key concept used in these proofs is the $Q$-path.

\begin{definition}[$Q$-path]
Given a vertex set $Q \subseteq V$, a $Q$-path is a simple path $v_0, v_1, \dots, v_n$ such that $v_i \in Q$ for all $i \in \{0, \dots, n-1\}$. The endpoint $v_n$ may or may not belong to $Q$.
\end{definition}

The invariants essentially state that during the inner loop initiated from vertex $x$:
\begin{itemize}
    \item For vertices $u$ already finalized ($u \in Q$), $\delta(u)$ correctly stores the shortest path distance $d(x,u)$ (Invariant \eqref{eq:delta-eq}).
    \item Vertices in $Q$ are always closer to $x$ (in terms of current $\delta$ values) than vertices not yet in $Q$ (Invariant \eqref{eq:delta-less}).
    \item The $\delta$ value for any vertex $v$ with a finite estimate corresponds to the length of a specific path found so far involving its predecessor (Invariant \eqref{eq:delta-pred}).
    \item $\delta(v)$ provides a lower bound for the length of any $Q$-path from $x$ to $v$ (Invariant \eqref{eq:deltaQpath1}).
    \item If $\delta(v)$ is finite, there exists a $Q$-path from $x$ to $v$ realizing this length (Invariant \eqref{eq:deltaQpath2}).
    \item The global variable $\gamma$ always remains a lower bound on the length of any cycle contained entirely within the currently explored region $Q$ and passing through $x$ (Invariant \eqref{eq:gammaC}).
    \item The variables $d^+_{\min}$, $d_{\text{to\_cycle}}$, $\ell_{\text{best}}$ correctly track the minimum composite distance (Invariant \eqref{eq:dplus_min}).
\end{itemize}

\begin{lemma}[Inner Loop Invariants]
\label{lem:inner}
The {\bf while} loop in Algorithm~\ref{alg:MWC_Dijkstra}, for a fixed outer loop iteration $x$, maintains the following invariants:
\begin{align}
\label{eq:delta-eq}
&\forall u \in Q: \delta(u) = d(x,u) \quad (\text{which is finite}) \\
\label{eq:delta-less}
&\forall u \in Q, \forall v \notin Q: \delta(u) \leq \delta(v) \\
\label{eq:delta-pred}
&\forall v \in V \setminus \{x\} \text{ with } \delta(v) < \infty: \exists u = \texttt{pred}(v) \in Q \text{ s.t. } \delta(v) = \delta(u) + w(u,v) = d(x,u) + w(u,v) \\
\label{eq:deltaQpath1}
&\forall v \in V, \forall Q\text{-paths } \pi \text{ from } x \text{ to } v: \delta(v) \leq \ell(\pi) \\
\label{eq:deltaQpath2}
&\forall v \in V \text{ with } \delta(v) < \infty: \exists \text{ a } Q\text{-path } \pi \text{ from } x \text{ to } v \text{ with } \delta(v) = \ell(\pi) \\
\label{eq:gammaC}
&\forall \text{ cycles } C \text{ with } x \in \text{vertices}(C) \subseteq Q: \gamma \leq \ell(C) \\
\label{eq:dplus_min}
&d^+_{\min} = \min_{c \text{ detected so far from } x} d^+(x,c) \text{ with } d_{\text{to\_cycle}}, \ell_{\text{best}} \text{ corresponding to a minimizing cycle}
\end{align}
\end{lemma}
\begin{proof} The proof proceeds by induction.
\textit{Base Case:} Initially $Q = \emptyset$, $\delta(x)=0$, $\delta(v)=\infty$ for $v \neq x$, $\gamma = \infty$, and $d^+_{\min} = \infty$. Invariants \eqref{eq:delta-eq}-\eqref{eq:delta-pred}, \eqref{eq:gammaC}, and \eqref{eq:dplus_min} hold vacuously or trivially. \eqref{eq:deltaQpath1} holds because the only $Q$-path is $x$ itself. Invariant \eqref{eq:deltaQpath2} holds since we need to consider only $v=x$.

\textit{Inductive Step:} Assume the invariants hold before an iteration where vertex $y$ is selected and added to $Q$. We need to show they hold after the iteration (when $Q' = Q \cup \{y\}$ and $\delta, \gamma, d^+_{\min}$ may be updated).
\begin{itemize}
    \item Invariant \eqref{eq:delta-eq}: We must show $\delta(y) = d(x,y)$, where ``$\ge$" follows from invariant (\ref{eq:deltaQpath2}); we shall use standard Dijkstra logic to show that it is impossible that $\delta(y) > d(x,y)$. For then there exists a path $\pi$ from $x$ to $y$ with $l(\pi) < \delta(y)$. Let $z$ be the first vertex on $\pi$ not in $Q$. Then $d(x,z) < \delta(y)$, but since $z \notin Q$, $\delta(z) \ge \delta(y)$ by invariant \eqref{eq:delta-less}, a contradiction since by invariant (\ref{eq:deltaQpath1}), $\delta(z) \le l(\pi) < \delta(y)$.
    \item Invariant \eqref{eq:delta-less}: If $u \in Q'$ and $v \notin Q'$, we need $\delta'(u) \le \delta'(v)$. If $u \in Q$, $\delta'(u)=\delta(u) \le \delta(y)$. If $u=y$, $\delta'(u)=\delta(y)$. If $v$ was updated, $\delta'(v) = \delta(y) + w(y,v) \ge \delta(y)$. If $v$ was not updated, $\delta'(v)=\delta(v) \ge \delta(y)$. The inequality holds.
    \item Invariant \eqref{eq:delta-pred}: Holds by construction for vertices $z$ whose $\delta$ value is updated in this iteration (Line 16). Uses the proven $\delta(y)=d(x,y)$.
    \item Invariants \eqref{eq:deltaQpath1}, \eqref{eq:deltaQpath2}: These proofs follow the standard Dijkstra arguments, considering paths that either do or do not pass through the newly added vertex $y$.
    \item Invariant \eqref{eq:gammaC}: If a cycle $C$ has $\text{vertices}(C) \subseteq Q'$, and it does not contain $y$, the invariant holds by induction. If $y \in C$, let $z_1, z_2$ be its neighbors in $C$. Since the entire cycle $C$ is contained within $Q' = Q \cup \{y\}$, and neither $z_1$ nor $z_2$ is the vertex $y$, it follows that both neighbors must have already been in the set $Q$ when $y$ was processed.
    
    We can without loss of generality assume that $z_1$ is processed from $y$ before $z_2$ is. Let $z$ be $z_1$ if $z_1 \neq \texttt{pred}(y)$, but let $z$ be $z_2$ if $z_1 = \texttt{pred}(y)$; in either case, $z \neq \texttt{pred}(y)$. We observe that processing $z$ causes a cycle $C'$ to be detected,
    formed by the edge $(y, z)$ and the paths in the current shortest path tree connecting $y$ and $z$ back to their Lowest Common Ancestor, $p$. Its length is $\ell(C') = w(y, z) + \delta(y) + \delta(z) - 2\delta(p)$.
    Our arbitrary cycle $C$ also contains the edge $(y, z)$. By the condition of the invariant, $C$ must pass through the root $x$. Therefore, the path connecting $y$ and $z$ within $C$ (excluding the edge $(y, z)$) must have a length of at least the shortest distance from $y$ to $x$ plus $x$ to $z_1$, which is $\delta(y) + \delta(z)$. Thus, $\ell(C) \ge w(y, z) + \delta(y) + \delta(z)$. Since edge weights are non-negative, $\delta(p) \ge 0$. Comparing the lengths, it must be that $\ell(C') \le \ell(C)$.
    The algorithm's update rule, $\gamma' \leftarrow \min(\gamma, \ell(C'))$, ensures that $\gamma' \le \ell(C')$. Combining these inequalities gives $\gamma' \le \ell(C') \le \ell(C)$, which proves the invariant holds.
    \item Invariant \eqref{eq:dplus_min}: Holds by construction (Lines 24-26). When a new cycle $c$ is detected, $d^+(x,c) = d_{x,c} + \ell_c = \delta(p) + [\delta(y) + \delta(z) + w(y,z) - 2\delta(p)]$ is computed, and $d^+_{\min}$ is updated if this value is smaller.
\end{itemize}
\end{proof}

The outer loop invariants ensure that the algorithm correctly maintains the global minimum cycle length estimate $\gamma$ and that the early termination condition of the inner loop is sound.

\begin{lemma}[Outer Loop Invariants]
\label{lem:outer}
The {\bf for} loop (Lines 3-end) in Algorithm~\ref{alg:MWC_Dijkstra} maintains the following invariants, where $X$ is the set of vertices $x$ processed so far by the outer loop:
\begin{align}
\label{eq:gamma-leq}
&\text{If } C \text{ is a cycle with } \text{vertices}(C) \cap X \neq \emptyset \text{ then } \gamma \leq \ell(C). \\
\label{eq:gamma-eq}
&\text{If } \gamma < \infty \text{ then there exists a cycle } C \text{ with } \ell(C) = \gamma. \\
\label{eq:active_correct}
&\text{All vertices of any MWC } c^* \text{ remain in } V_{\text{active}}.
\end{align}
\end{lemma}
\begin{proof}
(Sketch) Again, proved by induction on the iterations of the outer loop.
\textit{Base Case:} Initially $X=\emptyset, \gamma=\infty, V_{\text{active}} = V$. All invariants hold vacuously or trivially.

\textit{Inductive Step:} Assume invariants hold before processing vertex $x$. Let $\gamma_{old}$ be the value before the inner loop, $\gamma_{new}$ the value after. We know $\gamma_{new} \le \gamma_{old}$.
\begin{itemize}
    \item Invariant \eqref{eq:gamma-leq}: Consider a cycle $C$ with $\text{vertices}(C) \cap (X \cup \{x\}) \neq \emptyset$. If $\text{vertices}(C) \cap X \neq \emptyset$, then $\gamma_{old} \le \ell(C)$ by induction. Since $\gamma_{new} \le \gamma_{old}$, we have $\gamma_{new} \le \ell(C)$. Now assume $\text{vertices}(C) \cap X = \emptyset$, which means $x \in \text{vertices}(C)$. We need to show $\gamma_{new} \le \ell(C)$.
    Consider the state when the inner loop for $x$ terminates. Let $Q_{final}$ be the final set $Q$. The termination condition $\delta(v) < \gamma_{new}/2$ fails for all $v \notin Q_{final}$. This implies $\delta(v) \ge \gamma_{new}/2$ for all $v \notin Q_{final}$. 
        The true shortest path distance must also satisfy this bound: 
    \[
    d(x,v) \ge \gamma_{new}/2 \mbox{ for all } v \notin Q_{final}.
    \]

    To see this, assume that $d(x,v) < \gamma_{new}/2$ for some $v \notin Q_{final}$. Consider a shortest path from $x$ to $v$, and let $u$ be the first node on that path not in $Q_{final}$. Let $z$ be its predecessor on this path ($z \in Q_{final}$). When $z$ was processed and added to $Q$, the edge $(z,u)$ was relaxed, ensuring that $\delta(u) \le \delta(z) + w(z,u)$. Since this is a shortest path and $\delta(z)=d(x,z)$ (by Invariant~\eqref{eq:delta-eq}), we have $\delta(u) \le d(x,u)$. Furthermore, $d(x,u) \le d(x,v) < \gamma_{new}/2$. Thus, we have found a vertex $u \notin Q_{final}$ such that $\delta(u) < \gamma_{new}/2$. This contradicts the termination condition of the while loop (Line 9). Therefore, the assumption must be false, and $d(x,v) \ge \gamma_{new}/2$ for all $v \notin Q_{final}$.
    
    If $\text{vertices}(C) \subseteq Q_{final}$, then invariant \eqref{eq:gammaC} from Lemma \ref{lem:inner} ensures $\gamma_{new} \le \ell(C)$.
    If there exists $v \in \text{vertices}(C)$ such that $v \notin Q_{final}$, then we can write $C$ as two paths between $x$ and $v$, say $\pi_1$ and $\pi_2$. Then $\ell(C) = \ell(\pi_1) + \ell(\pi_2) \ge d(x,v) + d(x,v) = 2d(x,v)$. Since $v \notin Q_{final}$, $d(x,v) \ge \gamma_{new}/2$. Thus, $\ell(C) \ge 2(\gamma_{new}/2) = \gamma_{new}$. In all cases, $\gamma_{new} \le \ell(C)$. This shows the termination condition is sound.
    \item Invariant \eqref{eq:gamma-eq}: If $\gamma$ is updated in the inner loop (Line \ref{line:gamma_update}) to $\gamma_{new} = \ell(C')$ for some cycle $C'$ formed by paths $x \leadsto y$, edge $(y,z)$, path $z \leadsto x$, then the invariant holds.  The construction of the cycle $C'$ involves paths traced back via $\texttt{pred}$ to the LCA $p$ and ensures $C'$ is a simple cycle whose length matches the calculated value $\delta(y) + \delta(z) + w(y,z) - 2\delta(p)$. 
    
    If $\gamma$ is not updated, $\gamma_{new} = \gamma_{old}$, and the invariant holds by induction.
    \item Invariant \eqref{eq:active_correct}: This follows from Lemma~\ref{lem:vertex_discarding} below.
\end{itemize}
\end{proof}

The following lemma establishes that the vertex discarding step preserves correctness.

\begin{lemma}[Correctness of Vertex Discarding]
\label{lem:vertex_discarding}
The vertex discarding step (Lines 30-36 in Algorithm~\ref{alg:MWC_Dijkstra}) preserves correctness: no vertex belonging to any MWC is ever removed from $V_{\text{active}}$.
\end{lemma}
\begin{proof}
Consider the state after the inner loop for vertex $x$ completes. Let $c$ be the cycle achieving $d^+_{\min}$. The discarding logic is activated only if three conditions are met (Line 30):
\begin{enumerate}
    \item $d^+_{\min} < \infty$ (a cycle was detected),
    \item $\ell_{\text{best}} > \gamma$ (the cycle $c$ is not an MWC),
    \item $d^+_{\min} < 3\gamma/2$.
\end{enumerate}

We first show that under these conditions, $d^+_{\min}$ is the true minimum composite distance $d^+(x)$. The inner loop terminates when all vertices $v \notin Q$ satisfy $\delta(v) \ge \gamma/2$. Any cycle $c'$ not detected by the algorithm must have $d(x,c') \ge \gamma/2$. Furthermore, $\ell(c') \ge \gamma$. Thus, the composite distance of any undetected cycle is $d^+(x,c') = d(x,c') + \ell(c') \ge \gamma/2 + \gamma = 3\gamma/2$.

Since condition (3) ensures $d^+_{\min} < 3\gamma/2$, $d^+_{\min}$ must be strictly less than the composite distance of any undetected cycle. Therefore, $d^+_{\min} = d^+(x)$.

Now we apply the logic of Theorem~\ref{thm:exclude_vertices}. Let $c^*$ be any MWC. Since $d^+(x) = d^+(x,c)$ and $\ell(c) > \ell(c^*)$ (by condition 2):
\begin{align*}
d(x,c) + \ell(c) = d^+(x) &\leq d^+(x,c^*) = d(x,c^*) + \ell(c^*) \\
&< d(x,c^*) + \ell(c).
\end{align*}
Canceling $\ell(c)$ yields $d(x,c) < d(x,c^*)$.

Consider a vertex $z$ satisfying the discarding criterion (Line 32): $\delta(z) \leq d_{\text{to\_cycle}} = d(x,c)$. We know $d(x,z) \leq \delta(z)$ by invariant~\eqref{eq:deltaQpath1}.
If $z$ were on $c^*$, then $d(x,c^*) \leq d(x,z)$. This leads to the contradiction:
\[
d(x,c^*) \leq d(x,z) \leq \delta(z) \leq d(x,c).
\]
Therefore, $z \notin \text{vertices}(c^*)$, and discarding $z$ is safe.
\end{proof}

Finally, the overall correctness follows directly from the outer loop invariants holding upon termination.

\begin{theorem}[Overall Correctness]
\label{thm:overall_correctness}
Algorithm~\ref{alg:MWC_Dijkstra} terminates and returns $\gamma = \ell(c^*)$, where $c^*$ is a Minimum Weight Cycle in $G$.
\end{theorem}
\begin{proof}
The algorithm terminates because the outer loop iterates over $V_{\text{active}}$, a finite set that can only shrink, and the inner `while' loop processes each vertex at most once per outer iteration.

By invariant \eqref{eq:active_correct}, all vertices of any MWC $c^*$ remain in $V_{\text{active}}$ throughout execution. Therefore, the outer loop will eventually process some vertex $x^* \in \text{vertices}(c^*)$.

When vertex $x^* \in \text{vertices}(c^*)$ is processed, since $d(x^*, c^*) = 0$, the cycle $c^*$ will be detected (assuming it wasn't found earlier), and $\gamma$ will be updated to at most $\ell(c^*)$.

By invariant \eqref{eq:gamma-leq} of Lemma \ref{lem:outer}, $\gamma \leq \ell(C)$ for all cycles $C$ that share a vertex with any processed vertex. Since all MWC vertices are processed (they remain in $V_{\text{active}}$), we have $\gamma \le \ell(c^*)$.

By invariant \eqref{eq:gamma-eq}, there exists a cycle $C'$ with $\ell(C') = \gamma$. Since $c^*$ is an MWC, $\ell(C') \ge \ell(c^*)$. Therefore, $\gamma = \ell(c^*)$.
\end{proof}

\begin{remark}[Practical Implementation of Vertex Discarding]
\label{rem:practical_discarding}
The vertex discarding step (Lines 30-36) only considers vertices $z$ that were explored during the inner loop (i.e., those with finite $\delta(z)$ values). For vertices that were never reached (due to the $\gamma/2$ termination), their distance from $x$ exceeds $\gamma/2$. If $d_{\text{to\_cycle}} < \gamma/2$, such unexplored vertices automatically satisfy $d(x,z) > d_{\text{to\_cycle}}$ and do not meet the discarding criterion. This means the implementation efficiently checks only explored vertices for potential discarding.
\end{remark}

%==========%==========%==========%==========

\section{Complexity Analysis and Performance Enhancement}
\label{sec:complexity}

Having presented Algorithm \ref{alg:MWC_Dijkstra} and established its correctness, we now turn to its computational complexity and discuss strategies for practical performance enhancement, particularly the graph pruning heuristic.

\subsection{Worst-Case Complexity Analysis}
\label{subsec:worst_case}

Algorithm \ref{alg:MWC_Dijkstra} consists of an outer loop iterating over $V_{\text{active}}$, which initially contains all $|V|$ vertices. The inner `while' loop executes a modified Dijkstra search.
\begin{itemize}
    \item \textbf{Inner Loop (Modified Dijkstra):} In the worst case, without the $\gamma/2$ optimization significantly curtailing the search, the inner loop resembles a standard Dijkstra execution. Using a Fibonacci heap for the priority queue (to implement the $\arg\min$ in Line 10), selecting the minimum vertex takes amortized $O(\log |V|)$ time, and edge relaxations (Lines 13-17) take amortized $O(1)$ time per edge. Thus, one inner loop iteration runs in $O(|E| + |V|\log |V|)$ time.
    \item \textbf{LCA Queries:} The cycle detection step (Line 20) requires computing the Lowest Common Ancestor (LCA) within the current shortest path tree. This check can happen up to $O(|E|)$ times per inner loop. Using efficient online LCA data structures \cite{OnlineLCA}, each query can often be answered in $O(\log |V|)$ or even faster time (approaching constant time in practice after initial setup per tree), contributing an additional $O(|E| \log |V|)$ factor per inner loop in a straightforward analysis, though this might be pessimistic as tree structures evolve.
    \item \textbf{Overall Naive Bound:} Combining these, executing the inner loop for all $|V|$ starting vertices gives a naive worst-case complexity bound of approximately $O(|V|(|E| + |V|\log |V| + |E|\log |V|))$, which simplifies to $O(|V||E|\log|V| + |V|^2\log |V|)$. Assuming the graph is connected (i.e., $|V| = O(|E|)$), this further simplifies to $O(|V||E|\log|V|)$.\footnote{Alternatively, the complexity is often stated as $O(|V||E| + |V|^2\log |V|)$ if a Fibonacci heap is used for Dijkstra's algorithm and LCA costs are considered near-constant. This is a reasonable assumption, as advanced data structures allow for $O(1)$ (constant time) LCA queries after initial preprocessing \cite{OnlineLCA}. In this scenario, the total cost of LCA operations is ``absorbed,'' meaning it is asymptotically dominated by the cost of the shortest path computations.}
\end{itemize}

However, this analysis ignores the crucial optimizations:
\begin{itemize}
    \item \textbf{Early Termination ($\gamma/2$):} The condition $\delta(v) < \gamma/2$ (Line 9) can significantly prune the search space, especially once a relatively short cycle $\gamma$ is found. In graphs with short girth, many inner loops might terminate very quickly. Quantifying this speedup analytically for the general case is difficult, as it depends heavily on the graph structure and weight distribution.
    \item \textbf{Vertex Discarding (Theorem \ref{thm:exclude_vertices}):} Algorithm~\ref{alg:MWC_Dijkstra} explicitly integrates this optimization (Lines 30-36), including the necessary condition ($d^+_{\min} < 3\gamma/2$) to ensure safety when combined with the $\gamma/2$ early termination. This requires tracking the minimum composite distance $d^+_{\min}$ and $d_{\text{to\_cycle}}$, adding $O(1)$ overhead per cycle detection. After each inner loop, the discarding step iterates over explored vertices (at most $|V|$), adding $O(|V|)$ per outer iteration. This overhead is dominated by the Dijkstra search cost. The benefit is that subsequent outer loop iterations operate on a potentially smaller vertex set $V_{\text{active}}$, reducing both the number of iterations and the cost per iteration.
\end{itemize}

\begin{wrapfigure}[23]{r}{0.4\textwidth}
\vspace{-14pt}
	\centering
	\includegraphics[clip,width=0.38\textwidth]{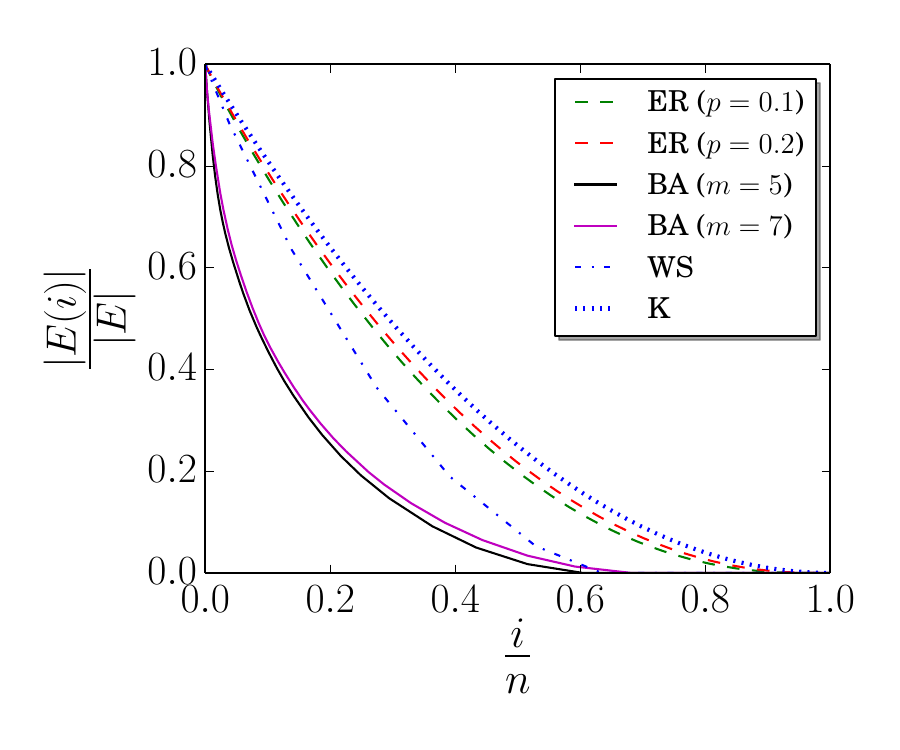}
    \vspace{-18pt}
	\caption{Fraction of the total possible edges examined ($|E_i|/|E|$) if vertices are removed one by one, sorted by degree (highest first), simulating an upper bound on the effect of vertex discarding on the worst case term $F$. Models shown: Erdős-Rényi (ER), Barabasi-Albert (BA), Watts-Strogatz (WS), Complete graph (K). The area under the curve relative to 1 indicates the potential reduction in edge processing compared to the naive $|V||E|$ term.}
    \label{fig:F_factor}
\end{wrapfigure}

Let $V_i$ and $E_i$ denote the set of active vertices and edges at the start of the $i$-th outer loop iteration (after potential removals from previous iterations). The complexity is better represented as $\sum_{i=1}^{|V|} O(|E_i| + |V_i|\log |V_i|)$, incorporating LCA costs within this bound. The effectiveness of vertex discarding determines how quickly $|V_i|$ and $|E_i|$ decrease. 

We can model the impact of vertex discarding by considering the total number of edge relaxations across all iterations. Let $F$ be the sum of degrees of all vertices over all iterations, accounting for removals:
\begin{equation}\label{eq:F_term}
F = \sum_{i=1}^{|V|} \sum_{j \in V_i} \text{deg}_{G_i}(j) \approx \sum_{i=1}^{|V|} 2|E_i|
\end{equation}
where $G_i=(V_i, E_i)$ is the graph at iteration $i$. The complexity related to edge processing then becomes roughly $O(F)$, and the part related to priority queue operations is $\sum_{i=1}^{|V|} O(|V_i|\log |V_i|)$. The overall complexity can be expressed as $O(F + |V|^2\log |V|)$ in the worst case for the priority queue over all iterations, though often $O(F + |V|\log|V| \cdot |V|)$ is used. The value of $F$ depends critically on the graph structure and the order of vertex processing (which affects discarding). As shown in Figure \ref{fig:F_factor}, $F$ can be significantly smaller than the naive bound $|V| \cdot 2|E|$ for many graph types, especially sparse ones or when high-degree vertices are processed early.
While vertex discarding offers a theoretical improvement, its practical impact varies. For further, more consistent speedups, especially when vertex discarding is not highly effective, we introduce a heuristic pruning method.

\subsection{Performance Enhancement: Graph Pruning Heuristic}
\label{subsec:pruning_heuristic_complexity}

Inspired by the locality principle and optimizations used in related iterative graph algorithms (like Loop Modulus computation in Section \ref{sec:loop_modulus_app}), we propose a heuristic graph pruning strategy to be used \textit{in conjunction} with Algorithm \ref{alg:MWC_Dijkstra}. This heuristic is distinct from the provable vertex discarding of Theorem \ref{thm:exclude_vertices}.

\textbf{Motivation:} Often, the shortest cycle, or cycles that are candidates for improving the current best $\gamma$, may be located structurally ``near'' the cycle that last updated $\gamma$. Continuously searching the entire (remaining) graph in every inner loop might be inefficient.
\textbf{Mechanism:} i) After the inner loop for vertex $x$ completes and potentially updates $\gamma$ based on a detected cycle $c'$, identify the vertices involved in $c'$. ii) For a limited number of subsequent outer loop iterations (controlled by a parameter $\texttt{pruning\_reset\_interval}$), restrict the \textit{next} inner loop's search (e.g., starting from $x+1$) to a subgraph $G_{\text{pruned}}$. iii) $G_{\text{pruned}}$ is constructed by taking vertices within a certain hop distance (parameter $\text{pruning\_distance\_threshold}$) from the vertices of $c'$ in the original graph structure. iv) The inner loop (Lines 9-28) then operates primarily within $G_{\text{pruned}}$: $\arg\min$ considers only vertices in $G_{\text{pruned}}$, and edge relaxation explores only edges induced by $G_{\text{pruned}}$. v) Crucially, after $\texttt{pruning\_reset\_interval}$ iterations, or if $\gamma$ fails to improve within the pruned view, the algorithm reverts to searching the full (remaining) graph $G_i$ to ensure global correctness.

\begin{figure}[htbp]
    \centering
    % Left subfigure
    \begin{minipage}[c]{0.32\textwidth}
        \centering
        \includegraphics[width=\textwidth, trim={0 3cm .57cm 3.5cm},clip]{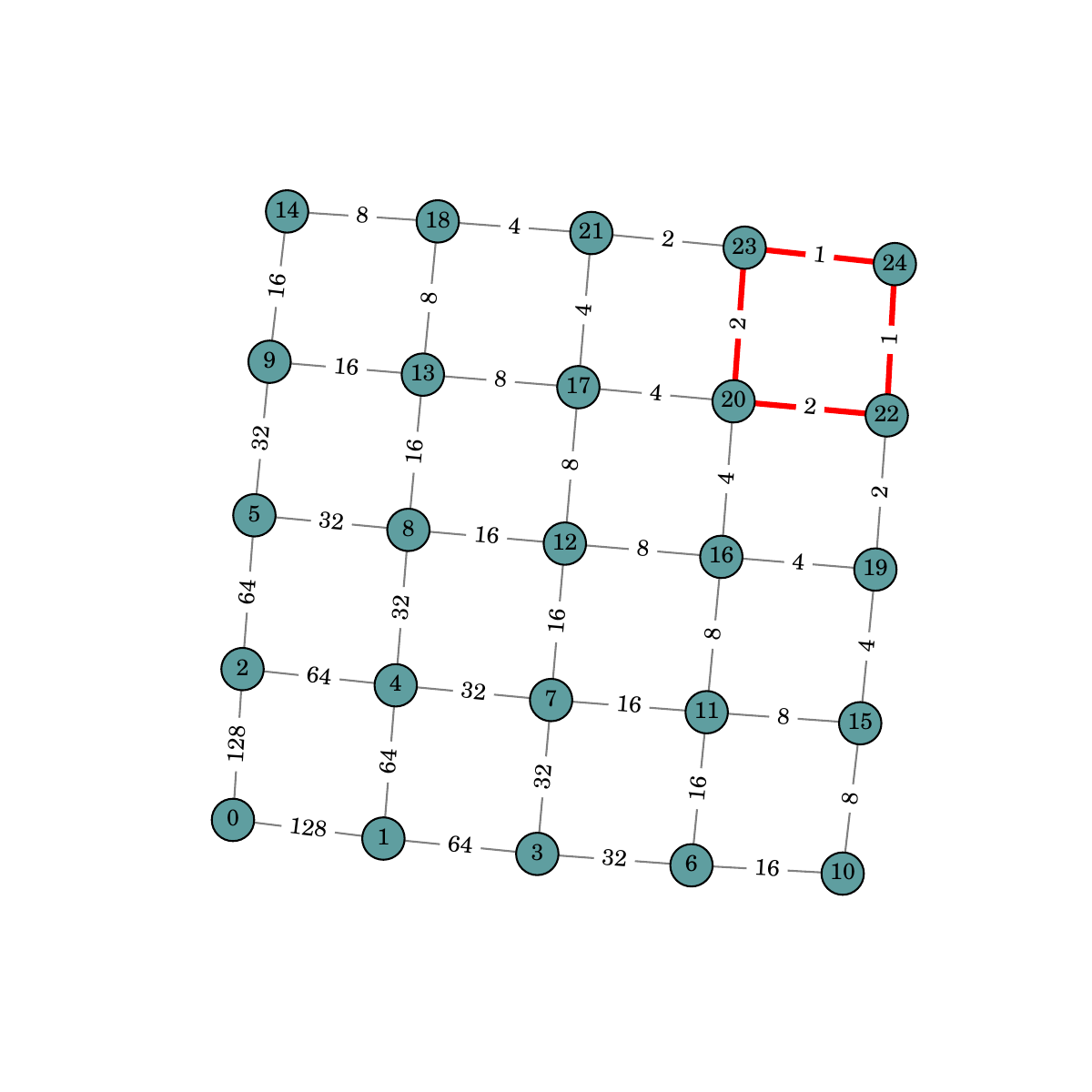}
    \end{minipage}%
    \hfill
    % Middle subfigure
    \begin{minipage}[c]{0.35\textwidth}
        \centering
        \includegraphics[width=\textwidth, trim={0 .5cm .8cm .85cm},clip]{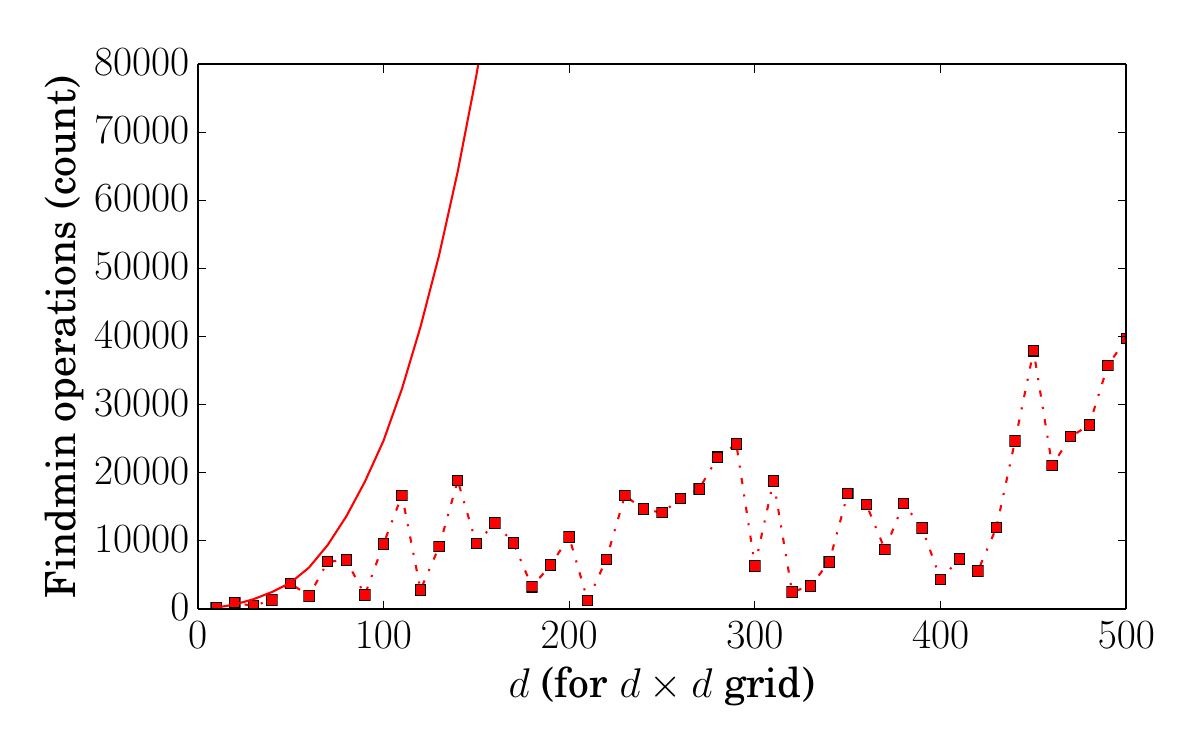}
    \end{minipage}%
    \hfill
    % Right subfigure
    \begin{minipage}[c]{0.28\textwidth}
        \centering
        \includegraphics[width=\textwidth, trim={25cm 10cm 10cm 30cm},clip]{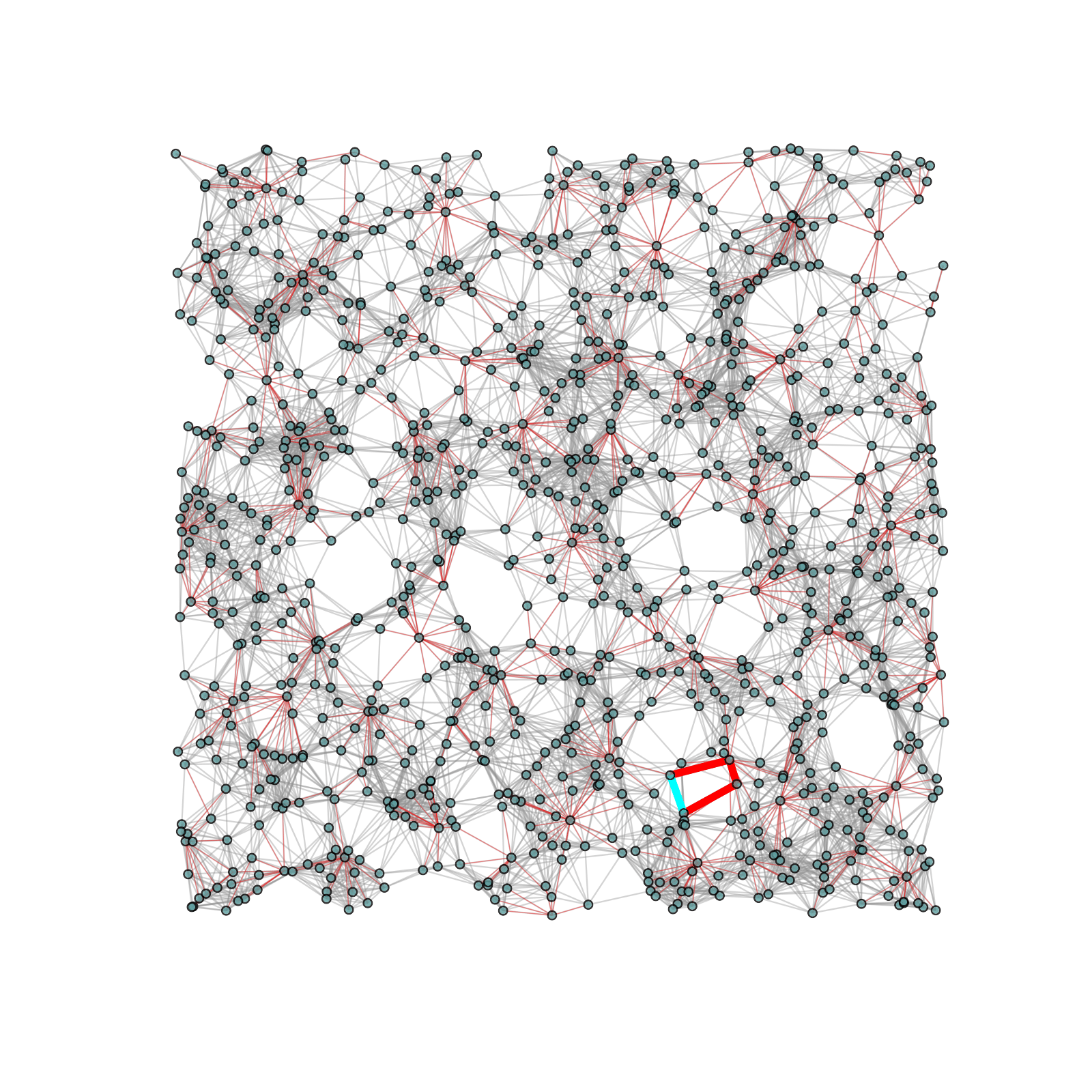}
    \end{minipage}
    \caption{(left) A $5\times 5$ weighted grid such that the MWC is localized near the highest labeled vertex. (middle) Comparison of the number of $\arg\min$ operations performed by the proposed Algorithm~\ref{alg:MWC_Dijkstra} (dashed line) versus the rooted girth algorithm (solid line) to find the MWC on $d \times d$ grids. (right) A shortest cycle (thick blue/cyan edges) formed by adding one light non-tree edge (cyan) to a light spanning tree (red edges) within a larger spatial graph.}
    \label{fig:grid}
\end{figure}

\textbf{Impact:} This heuristic does not improve the theoretical worst-case complexity (which is determined by the full graph searches during resets), but it aims to drastically reduce the \textit{empirical runtime}. By operating on a potentially much smaller $|V_{\text{pruned}}|$ and $|E_{\text{pruned}}|$ for many inner loop iterations, the cost $O(|E_{\text{pruned}}| + |V_{\text{pruned}}|\log |V_{\text{pruned}}|)$ per iteration can be significantly lower. The overhead involves the BFS to determine $G_{\text{pruned}}$, which is typically less costly than the potential savings in the Dijkstra search.
The combination of the $\gamma/2$ termination, provable vertex discarding, and the optional graph pruning heuristic makes Algorithm \ref{alg:MWC_Dijkstra} significantly faster in practice than naive MWC algorithms, as illustrated by the following examples.

\subsection{Comparison Point: The Rooted Girth Algorithm}
\label{subsec:rooted_girth}

To contextualize the performance of Algorithm \ref{alg:MWC_Dijkstra}, we compare it to the standard Rooted Girth Algorithm introduced in Section \ref{sec:intro_framework}. Recall that this approach requires $|E|$ shortest path computations by iterating through all edges $e = (u,v) \in E$ and computing $d_{G\setminus e}(u,v)$.

While conceptually simple, this approach is generally significantly less efficient than Algorithm \ref{alg:MWC_Dijkstra}. The Rooted Girth algorithm must execute all $|E|$ computations, whereas the optimizations in Algorithm \ref{alg:MWC_Dijkstra} (particularly the $\gamma/2$ cutoff and vertex discarding) frequently allow many searches to terminate early, leading to substantial empirical speedups. This difference is particularly pronounced on graphs where the MWC is long or involves vertices far from each other.

\subsection{Example 1: Grid Graph with a Localized Short Cycle}
\label{subsec:example_grid}

Consider a $d \times d$ grid graph, where vertices are labeled starting from $0$ at $[0,0]$ via BFS. We assign edge weights exponentially increasing with distance from the highest-labeled vertex $v_{last}$ (at $[d-1, d-1]$): edges incident to $v_{last}$ have weight $2^0=1$, edges incident to vertices at hop distance $h$ from $v_{last}$ have weight $2^h$. An example $5 \times 5$ grid is shown in Figure \ref{fig:grid} (left).

The unique MWC in this construction is the square of weight $1+1+2+2 = 6$ incident to $v_{last}$. For the rooted girth algorithm starting searches from low-labeled vertices, reaching this cycle requires exploring a large portion of the graph. Algorithm \ref{alg:MWC_Dijkstra}, however, benefits from its structure. While starting from $x=0$ might be slow, once an outer loop starts from a vertex $x$ closer to the MWC, the $\gamma/2$ condition and potential vertex discarding can accelerate finding and verifying the MWC. Figure \ref{fig:grid} (middle) empirically compares the number of $\arg\min$ operations (a proxy for Dijkstra cost) required by both approaches, showing Algorithm \ref{alg:MWC_Dijkstra} (dashed line) requires significantly fewer operations, especially for larger grids, finding the MWC much faster (often within processing just a few high-labeled vertices).

\subsection{Example 2: Light Cycle on a Spanning Tree}
\label{subsec:example_spanning_tree}

Consider a base graph (e.g., a random geometric graph) where edge weights are initially large. We find an arbitrary spanning tree $T$ and re-weight all edges $e \in E(T)$ to $w(e)=1$. Then, we select one non-tree edge $e_{NT} = (u,v) \notin E(T)$ and set its weight $w(e_{NT})=1$. All other non-tree edges retain their large weights.

The unique MWC in this graph is formed by the edge $e_{NT}$ plus the unique path between $u$ and $v$ within the spanning tree $T$. Finding this cycle using the rooted girth algorithm can be inefficient, as it must test $|E|$ edges, many of which belong only to very long cycles. Algorithm \ref{alg:MWC_Dijkstra}, when run starting from any vertex $x$, will explore the low-weight spanning tree edges efficiently. When the Dijkstra search expands across the edge $e_{NT}$, it will likely quickly detect the short cycle involving the tree path (Figure \ref{fig:grid}(right)). The $\gamma/2$ condition will then likely terminate subsequent searches rapidly. This structure highlights cases where Algorithm \ref{alg:MWC_Dijkstra} can significantly outperform edge-rooted approaches.

\section{Application: Accelerating Loop Modulus Computation}
\label{sec:loop_modulus_app}

A fundamental problem in network analysis is quantifying the ``richness'' of cyclic structures within a graph. The $p$-Modulus of a family of loops $L$ provides such a measure, analogous to concepts in complex analysis \cite{albin2016minimal}. For $p=2$, which offers computational advantages and a useful probabilistic interpretation, the modulus is defined via the quadratic programming problem (QP) \cite{shakeri2017network}:

\begin{equation}
\label{eq:modulus_primal}
\text{Mod}_2(L) = \min_{\rho \ge 0} \sum_{e \in E} \rho(e)^2 \quad \text{subject to} \quad \sum_{e \in \gamma} \rho(e) \ge 1 \quad \forall \gamma \in L
\end{equation}
Here, $\rho: E \to \mathbb{R}_{\ge 0}$ is a density function on the edges. The optimal density $\rho^*$ minimizing this quadratic energy function subject to the constraints (where each simple loop $\gamma$ must have a total $\rho$-length of at least 1) provides valuable information. Specifically, $\rho^*(e)$ can be interpreted as proportional to the expected usage of edge $e$ by ``important'' loops within the family $L$. The final $\text{Mod}_2(L)$ value quantifies the overall cycle richness, balancing the number and length of loops against their overlap. This measure and the resulting $\rho^*$ densities have applications in network clustering, community detection, and understanding network robustness \cite{shakeri2017network}.

\subsection{The Iterative Modulus Algorithm and its Bottleneck}
\label{subsec:modulus_algo_bottleneck}

Directly solving the primal problem \eqref{eq:modulus_primal} is intractable due to the potentially enormous number of simple loops $L$. Practical algorithms, summarized in Algorithm \ref{alg:loop_mod}, employ an iterative approach based on constraint generation \cite{shakeri2017network}:
\begin{enumerate}
    \item \textbf{Initialization (Preprocessing):} Start with an empty or small initial set of loop constraints $L'$ (e.g., triangles found via heuristics or a shortest hop cycle, Lines 1-9).
    \item \textbf{Initial QP Solve:} Solve the modulus QP problem restricted to $L'$ to obtain an initial density $\rho^{(0)}$ and warm-start variables (Lines 10-14).
    \item \textbf{Constraint Finding (Bottleneck):} In iteration $k \ge 0$, find one or more simple cycles $\gamma$ \textit{not} currently in $L'$ that violate the constraints for the \textit{current} density $\rho^{(k)}$. That is, find $\gamma$ such that its $\rho$-length satisfies:
    \begin{equation}
    \label{eq:rho_length_violation}
    l_\rho(\gamma) = \sum_{e \in \gamma} \rho^{(k)}(e) < 1.0 - \epsilon_{\text{tol}}
    \end{equation}
    where $\epsilon_{\text{tol}}$ is a small positive tolerance. The cycle(s) with the \textit{minimum} $l_\rho$ are the ``most violated'' (Line 28, implemented by `FindTopKViolatedCycles`).
    \item \textbf{Add Constraint(s):} Add the found new unique violating cycle(s) to the set $L'$. If no new violating cycles are found, convergence is reached.
    \item \textbf{Re-solve QP:} Solve the QP with the updated constraint set $L'$, using warm-starting from the previous solution to accelerate the solver (Lines 41-44). Update $\rho^{(k+1)}$.
    \item \textbf{Iteration:} Repeat steps 3-5 until convergence or a maximum iteration limit is reached (Outer `while' loop, Line 16).
\end{enumerate}
The critical bottleneck in this process is Step 3, the \textbf{Constraint Finding}. This step requires searching the graph $G$ with edge weights defined by the \textit{current} density $\rho^{(k)}$ (which changes every iteration) to find the simple cycle(s) with the smallest $\rho$-length. Standard approaches often involve variants of repeated Dijkstra searches or related techniques, which can be computationally prohibitive \cite{shakeri2017network}.

\subsection{Using Algorithm \ref{alg:MWC_Dijkstra} for Efficient Constraint Finding}
\label{subsec:algo1_for_constraints}

The constraint-finding step (Step 3 above) is fundamentally a search for a minimum weight simple cycle (or top-k shortest cycles), where the edge weights are given by the dynamic density $\rho^{(k)}$. Algorithms specifically designed for finding the MWC, such as Algorithm \ref{alg:MWC_Dijkstra} presented in Section \ref{sec_alg} or others employing efficient techniques \cite{itai1978finding, orlin2016nm}, are prime candidates to implement the \texttt{FindTopKViolatedCycles} function.

In iteration $k$ of the modulus algorithm, the chosen MWC algorithm is executed using the current density $\rho^{(k)}$ as the edge weights $w(e) = \rho^{(k)}(e)$. Its goal is to efficiently identify at least one, and preferably up to $k_{\text{add}}$, simple cycles $\gamma$ satisfying Equation \eqref{eq:rho_length_violation}. The efficiency stems from exploiting the structure of the MWC problem itself, potentially avoiding the exhaustive nature of simpler APSP-based methods. For instance, Algorithm \ref{alg:MWC_Dijkstra}'s $\gamma/2$ pruning condition (adapted to use $1.0 - \epsilon_{\text{tol}}$ perhaps, or simply running until the first few cycles shorter than the threshold are found) can significantly limit the search effort required within the \texttt{FindTopKViolatedCycles} call.

\paragraph{Efficiency gains in iterative context.}  In the conventional baseline approach (e.g., using the Rooted Girth Algorithm), each Loop-Modulus iteration requires $|E|$ independent shortest-path searches, leading to a high per-iteration cost. Our MWC algorithm (Algorithm~\ref{alg:MWC_Dijkstra}) significantly reduces this per-iteration cost. Although it must be re-run each time the weights $\rho^{(k)}$ change, its internal optimizations (the $\gamma/2$ cutoff, vertex discarding, and graph pruning) often allow it to find the most violated cycle much faster than the baseline approach—typically one to two orders of magnitude faster. This substantial speedup in the constraint-finding bottleneck drastically reduces the overall runtime of the Loop Modulus computation. Further engineering accelerations such as BMSSP recursion \cite{duan2025breaking} or faster heap implementations can provide additional gains.

\subsection{Performance Enhancement: The Graph Pruning Heuristic}
\label{subsec:pruning_heuristic}

Even with an optimized MWC algorithm, searching the entire graph in every iteration can be redundant. Algorithm \ref{alg:loop_mod} incorporates an optional graph pruning heuristic (controlled by $P_{\text{use}}$) to focus the search:

\begin{itemize}
    \item \textbf{Motivation:} The most violated cycle(s) in iteration $k+1$ might often be structurally ``close" to the violating cycle(s) added in iteration $k$, as the $\rho$ adjustments primarily occur along those paths.
    \item \textbf{Mechanism (Lines 45-52):}
        \begin{enumerate}
            \item \textit{Trigger:} After solving the QP and identifying the nodes $V_{\text{last\_added}}$ involved in the newly added cycles $L'_{\text{new}}$ (Line 40), and if pruning is enabled ($P_{\text{use}}$) and the algorithm is not currently in a reset phase ($k_{\text{prune\_steps}} == 0$).
            \item \textit{Expand Region:} Perform a BFS starting from $V_{\text{last\_added}}$ on the original graph $G$ to find all nodes $V_{\text{keep}}$ within $P_{\text{dist}}$ hops (Line 46).
            \item \textit{Create/Check View:} Attempt to create a subgraph $G_{\text{view}} = G[V_{\text{keep}}]$ (Line 50). Check if the pruning is too aggressive (e.g., $|V_{\text{keep}}| < 0.3 |V|$). If yes, force a reset by setting $G_{\text{view}} \gets G$ and ensuring the next iteration searches the full graph (Line 48). Otherwise, set $G_{\text{view}}$ to the pruned subgraph and start the prune counter $k_{\text{prune\_steps}} \gets 1$ (Line 51).
            \item \textit{Localized Search (Lines 19-26):} In subsequent iterations, if pruning is active ($k_{\text{prune\_steps}} < P_{\text{interval}}$), the `FindTopKViolatedCycles' function (Line 28) is called with $G_{\text{search}} = G_{\text{view}}$. The MWC algorithm operates within this view, using global $\rho^{(k)}$ weights and validating found cycles against the full graph $G$.
            \item \textit{Reset:} The search automatically reverts to the full graph $G_{\text{search}} = G$ when $k_{\text{prune\_steps}}$ reaches $P_{\text{interval}}$ (Line 24), or if the pruning was deemed too aggressive (Line 48).
        \end{enumerate}
\end{itemize}

\begin{figure*}[htbp!]
    \centering
        \begin{algorithm}[H]
            \caption{Compact Iterative Loop Modulus Calculation (p=2)}
            \label{alg:loop_mod}
            \begin{algorithmic}[1]
            \algnotext{EndFor}
            \algnotext{EndIf}
            \algnotext{EndWhile}
                \Require Graph $G=(V, E)$, tol $\epsilon$, max iter $K$, init target $N_{\text{tgt}}$, cycles/iter $k_{\text{add}}$, prune params $P_{\text{use}}, P_{\text{int}}, P_{\text{dist}}$
                \Ensure Optimal density $\rho^*$, Modulus $M$

                \Statex \textit{// Phase 1: Preprocessing}
                \State $L_{\text{cand}} \gets \Call{FindTriangles}{G}$ \textbf{or} $\{\Call{FindShortestHopCycle}{G}\}$ \textbf{if} none
                \If{$L_{\text{cand}}$ is empty} \Return $\rho \gets \mathbf{0}$, $M \gets 0$ \EndIf
                \State Score $L_{\text{cand}}$ (e.g., edge centrality)
                \State $L' \gets \Call{GreedySelect}{L_{\text{cand}}, N_{\text{tgt}}}$
                \State $\rho \gets \mathbf{0}$; $k \gets 0$; $M \gets 0$; $k_{\text{QP}} \gets 0$
                \State $G_{\text{view}} \gets G$; $k_{\text{prune}} \gets 0$; $V_{\text{last}} \gets \emptyset$

                \Statex \textit{// Phase 2: Initial QP Solve}
                \State $N \gets \Call{BuildConstraintMatrix}{L'}$
                \State $\rho, \text{w\_x}, \text{w\_y} \gets \Call{SolveQP}{N, \text{None}}$ ; $M \gets \sum \rho^2$ ; $k_{\text{QP}} \gets 1$

                \Statex \textit{// Phase 3: Iterative Constraint Addition}
                \While{$k < K$}
                    \State $k \gets k + 1$
                    \If{$P_{\text{use}}$ \textbf{and} $k_{\text{prune}} < P_{\text{int}}$ \textbf{and} $G_{\text{view}} \neq G$} \label{line:prune_check_compact} \Comment{Use pruned view?}
                         \State $G_{\text{search}} \gets G_{\text{view}}$ ; $k_{\text{prune}} \gets k_{\text{prune}} + 1$
                    \Else \Comment{Use full graph / Reset}
                        \State $G_{\text{search}} \gets G$ ; $k_{\text{prune}} \gets 0$ ; $G_{\text{view}} \gets G$ \label{line:prune_logic_compact}
                    \EndIf

                    \State $\rho_{\text{map}} \gets \text{dict}(E \to \rho(e))$ \Comment{Find violating cycles}
                    \State $L_{\text{viol}} \gets \Call{FindTopKViolatedCycles}{G_{\text{search}}, G, \rho_{\text{map}}, k_{\text{add}}, 1 - \epsilon}$ \label{line:find_cycles_compact}
                    \If{$L_{\text{viol}}$ is empty} \textbf{break} \EndIf \Comment{Converged}

                    \State $L'_{\text{new}} \gets \{\gamma \mid (\gamma, l_\rho) \in L_{\text{viol}} \text{ and } \gamma \notin L'\}$ \Comment{Collect new unique constraints}
                    \If{$L'_{\text{new}}$ is empty} \textbf{break} \EndIf \Comment{No new violations found}
                    \State $L' \gets L' \cup L'_{\text{new}}$ ; $V_{\text{last}} \gets \text{nodes in } L'_{\text{new}}$

                    \State $N \gets \Call{BuildConstraintMatrix}{L'}$ \Comment{Re-solve QP}
                    \State $\rho, \text{w\_x}, \text{w\_y} \gets \Call{SolveQP}{N, (\text{w\_x}, \text{w\_y})}$
                    \State $M \gets \sum \rho^2$ ; $k_{\text{QP}} \gets k_{\text{QP}} + 1$

                    \If{$P_{\text{use}}$ \textbf{and} $k_{\text{prune}} == 0$ \textbf{and} $L'_{\text{new}}$ is not empty} \label{line:prune_update_compact} \Comment{Update pruning state}
                         \State $V_{\text{keep}} \gets \Call{BFSFromNodes}{G, V_{\text{last}}, P_{\text{dist}}}$
                         \If{$|V_{\text{keep}}| < 0.3 |V|$} $G_{\text{view}} \gets G$ ; $k_{\text{prune}} \gets P_{\text{int}}$ \Comment{Check if prune too aggressive}
                         \Else $G_{\text{view}} \gets G[V_{\text{keep}}]$ ; $k_{\text{prune}} \gets 1$
                         \EndIf
                    \EndIf
                \EndWhile
                \State $\rho^* \gets \rho$
                \State \Return $\rho^*$, $M$, $k_{\text{QP}}$
            \end{algorithmic}
        \end{algorithm}
\end{figure*}

\subsection{Correctness and Performance Impact}
\label{subsec:correctness_performance}

\begin{itemize}
    \item \textbf{Correctness:} The pruning heuristic does not compromise the theoretical correctness of the overall loop modulus algorithm (Algorithm \ref{alg:loop_mod}). The key is the \textit{periodic reset mechanism}. Even if the globally most violating cycle lies outside the pruned view $G_{\text{view}}$, the algorithm will eventually revert to searching the full graph $G$, guaranteeing that any violations necessary for convergence will be found.
    \item \textbf{Performance:} The primary benefit is empirical speedup. By restricting the MWC search (the implementation of `FindTopKViolatedCycles`) to the potentially much smaller $G_{\text{view}}$ for multiple iterations, the computational cost of the bottleneck step is significantly reduced. While the BFS for pruning (Line 46) adds overhead, it is often negligible compared to the savings achieved by running the MWC search on a smaller graph, especially when the MWC algorithm's complexity scales super-linearly with graph size.
\end{itemize}

\subsection{Example Performance: Cholera Dataset}
\label{subsec:cholera_example}

To illustrate the potential benefit, consider the application of an optimized iterative modulus algorithm to a real-world network dataset derived from the 1854 Broad Street cholera outbreak in London \cite{snow1856mode}. The graph (Figure \ref{fig:cholera_results_viz}(left)) can be constructed, for instance, using Delaunay triangulation based on the locations of cholera cases, representing geographic proximity relevant to potential water source usage \cite{networkx_delaunay}.
Applying an optimized algorithm (using OSQP, warm-starting, batch constraint addition, and efficient constraint finding incorporating principles similar to Algorithm 1 with pruning) yields significant performance improvements compared to less optimized approaches, as shown in Table \ref{tab:cholera_results}. The $\rho^*$ values obtained highlight edges connecting areas that frequently participate in loops, potentially indicating regions strongly associated with shared water sources (Figure \ref{fig:cholera_results_viz}).

\begin{table}[ht]
 \centering
 \begin{tabular}{|l | r | c | r | l |}
 \hline
 \textbf{Method} & \textbf{QP Solves} & \textbf{Total Time (s)} & \textbf{Final Modulus} & \textbf{Total Constraints}  \\ \hline 
 Our proposed algorithm & \textbf{28} & \textbf{3.8} & \textbf{100.8} & \textbf{248} \\ \hline
 Baseline Dijkstra & 475 & 700 & 99.1 & 475 \\ \hline
 \end{tabular}
 \caption{Illustrative performance comparison for Loop Modulus calculation on the Cholera dataset graph ($\sim$324 nodes, $\sim$941 edges).}
 \label{tab:cholera_results}
\end{table}

\begin{figure}
 \centering
    \begin{minipage}[c]{0.25\textwidth}
        \centering
        \includegraphics[width=\textwidth]{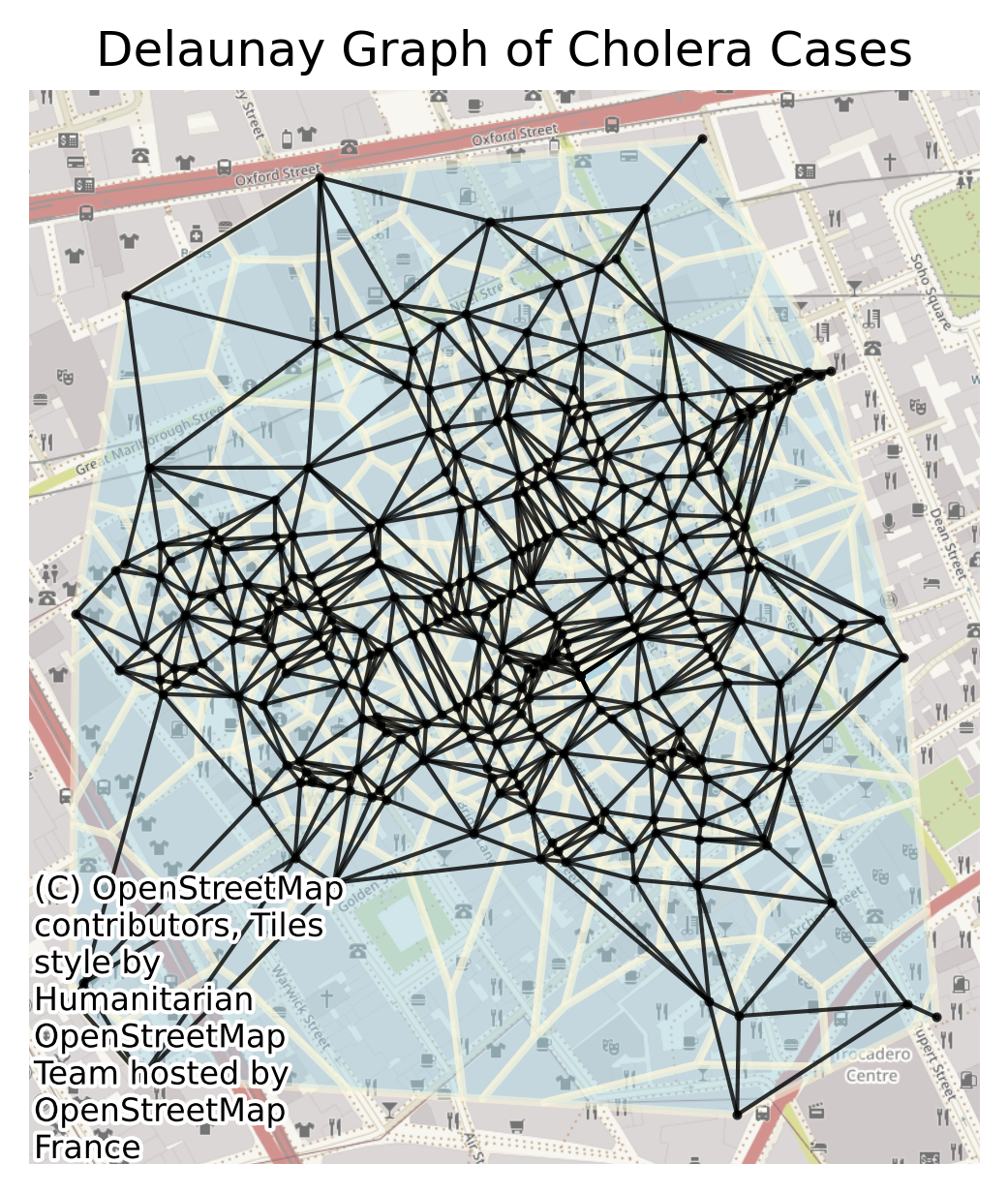}
    \end{minipage}%
    \hfill
    \begin{minipage}[c]{0.75\textwidth}
        \centering
        \includegraphics[width=\textwidth, trim={0 0 .57cm 2cm},clip]{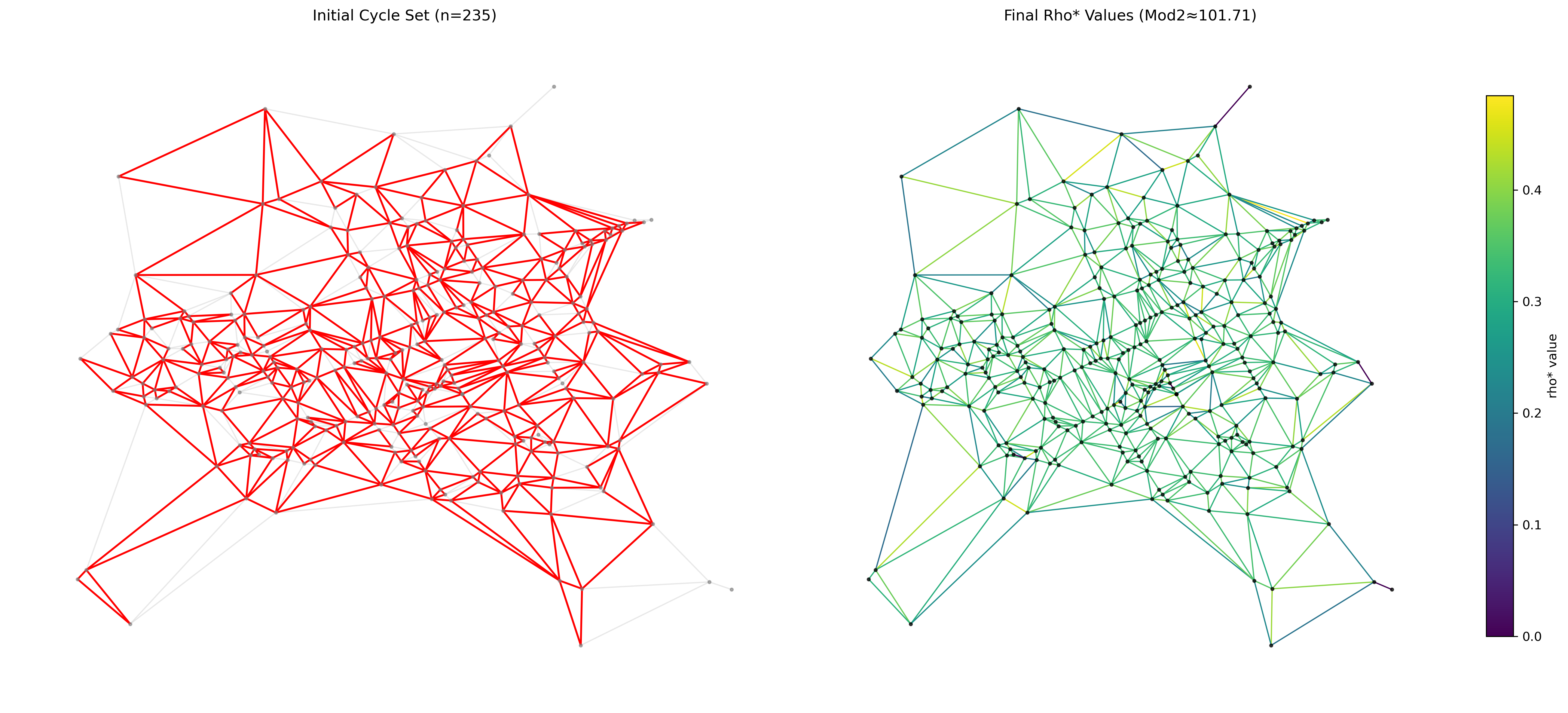}
    \end{minipage}
 \caption{(left) Graph structure derived from the 1854 Cholera outbreak data using Delaunay triangulation of case locations \cite{networkx_delaunay}. (middle) Visualization of Loop Modulus results ($\rho^*$ edge densities) on the Cholera graph. (right) Edge thickness/color intensity corresponds to higher $\rho^*$ values, highlighting edges frequently part of important loops.}
 \label{fig:cholera_results_viz}
\end{figure}

The results demonstrate a dramatic reduction in both the number of QP solves and the total runtime. This highlights the significant practical advantage gained by optimizing the constraint-finding bottleneck, where techniques derived from efficient MWC algorithms like Algorithm 1, coupled with heuristics like graph pruning, play a crucial role.

\section{Conclusions}
\label{sec:conclusion}

In this paper, we introduced a novel algorithm for identifying the Minimum Weight Cycle (MWC) in weighted graphs, a fundamental problem in graph theory with wide-ranging applications in network analysis and optimization. Our approach redefines the MWC search by minimizing a \textit{composite distance metric}, which integrates the shortest path distance from a vertex to a cycle with the cycle's own length. This transforms the traditionally global cycle search into an efficient, iterative, node-centric optimization process, drawing inspiration from Dijkstra's algorithm. We have substantiated the algorithm's correctness through rigorous proofs grounded in loop invariants, ensuring its reliability across diverse graph structures.
To enhance computational efficiency, we incorporated two key optimizations:
\begin{itemize}
    \item Node Discarding Technique: Leveraging intermediate results, this method reduces the search space by safely excluding vertices that cannot belong to the MWC, as supported by a formal theorem. The integration with the $\gamma/2$ early termination requires an additional safeguard ($d^+_{\min} < 3\gamma/2$) to ensure correctness.
    \item Graph Pruning Heuristic: This dynamic strategy focuses the search on relevant subgraphs, exploiting the locality principle prevalent in complex networks to achieve significant empirical speedups, while periodic resets preserve global optimality.
\end{itemize}
These optimizations contribute to the algorithm's practical efficiency, achieving a complexity of \(\mathcal{O}(|V|^2 \log |V| + F)\), where \(F\) is a graph-structure-dependent factor that can range from \(\mathcal{O}(nm)\) in dense or challenging cases to \(\mathcal{O}(1)\) in highly favorable scenarios, such as sparse graphs with effective pruning. We illustrate the utility of the algorithm by integrating it into the iterative constraint-finding process of Loop Modulus computation, where it substantially reduces runtime, as demonstrated in a case study using the Cholera dataset. This practical utility underscores the algorithm's value as a core primitive for advanced network analysis tasks, including clustering, community detection, and robustness assessment.

Our contributions advance the toolkit available for mining cyclic graph topologies by offering a fast, reliable, and theoretically sound solution. The composite distance approach not only improves upon traditional methods that exhaustively explore all cycles or edges but also adapts efficiently to real-world network structures. Looking ahead, potential extensions could include:
\begin{itemize}
    \item Parallelization: Distributing the vertex-centric searches across multiple processors to further reduce runtime.
    \item Dynamic Graphs: Adapting the algorithm to handle graphs with evolving edge weights.
    \item Directed or Negative-Weighted Graphs: Expanding the framework to directed graphs or those with negative weights (assuming no negative cycles), broadening its applicability.
\end{itemize}
As a next step to improve our algorithm for finding minimum weight cycles, we propose exploring future randomized versions inspired by probabilistic methods from graph theory. Drawing on the work of Albin and Poggi-Corradini \cite{albin2016minimal}, randomized sampling techniques based on the probabilistic interpretation of modulus could approximate the composite distance or edge importance, potentially yielding faster algorithms with provable approximation guarantees. Similarly, adapting the Renewal Non-Backtracking Random Walk \cite{moradi2021new} could enable efficient cycle sampling by prioritizing edges with high retracing probabilities. These randomized approaches promise enhanced scalability and efficiency, particularly for large-scale networks, building on the theoretical and practical insights from these studies.

\bibliographystyle{unsrtnat}
\bibliography{refs}

\end{document}